\documentclass[10pt]{bmc_article}    

\usepackage{cite} % Make references as [1-4], not [1,2,3,4]
\usepackage{url}  % Formatting web addresses  
\usepackage{ifthen}  % Conditional 
\usepackage{multicol}   %Columns
\usepackage{algorithm,algorithmic,graphicx,amsmath, amsthm, amssymb}    % manual
\usepackage[normalem]{ulem}                                             % manual
\usepackage{color}

\usepackage{caption}                                                    % manual

\newboolean{publ}
\newenvironment{bmcformat}{\fussy\setboolean{publ}{true}}{\fussy}

\renewcommand{\paragraph}[1]{\par\noindent\textbf{#1}}

\newcommand{\problem}[1]{\begingroup\bfseries\sffamily#1\endgroup}
\newtheorem{lemma}{Lemma}
\newtheorem{theorem}{Theorem}
\newtheorem{corollary}{Corollary}

\begin{document}
\begin{bmcformat}
\title{Maximizing Output and Recognizing Autocatalysis in\newline 
Chemical Reaction Networks is NP-Complete}

\author{Jakob Lykke Andersen$^{1}$%
  \email{jakan06@student.sdu.dk}%
  \and
  Christoph Flamm$^{2}$%
  \email{CF\correspondingauthor:xtof@tbi.univie.ac.at}%
  \and
  Daniel Merkle$^{1}$%
  \email{DM\correspondingauthor:daniel@imada.sdu.dk}%
  \and
  Peter F.\ Stadler$^{2-7}$%
  \email{PFS:studla@bioinf.uni-leipzig.de}%
}

\address{%
  \iid(1) Department for Mathematics and Computer Science, 
  University of Southern Denmark, Campusvej 55, DK-5230 Odense M, Denmark
  \iid(2) Institute for Theoretical Chemistry, University of Vienna, 
  W{\"a}hringerstra{\ss}e 17, A-1090 Wien, Austria.
  \iid(3) Bioinformatics Group, Department of Computer Science,
  and  Interdisciplinary Center for Bioinformatics,
  H{\"a}rtelstra{\ss}e 16-18, D-04107, Leipzig, Germany.
  \iid(4) Max Planck Institute for Mathematics in the Sciences,
  Inselstra{\ss}e 22 D-04103 Leipzig, Germany.
  \iid(5) Fraunhofer Institute for Cell Therapy and Immunology,
  Perlickstra{\ss}e 1, D-04103 Leipzig, Germany.
  \iid(6) Center for non-coding RNA in Technology and Health,
  University of Copenhagen, Gr{\o}nneg{\aa}rdsvej 3, DK-1870 
  Frederiksberg C, Denmark.
  \iid(7) Santa Fe Institute, 1399 Hyde Park Rd, Santa Fe, NM 87501, USA
}%

\maketitle

%%%%%%%%%%%%%%%%%%%%%%%%%%%%%%%%%%%%%%%%%%%%%%%%%%%%%%%%%%%%%%%%%%%%%%%%%%%%

%\newpage
\begin{abstract}
  \paragraph{Background:} 
  A classical problem in metabolic design is to maximize the production of
  desired compound in a given chemical reaction network by appropriately
  directing the mass flow through the network. Computationally, this
  problem is addressed as a linear optimization problem over the ``flux
  cone''. The prior construction of the flux cone is computationally
  expensive and no polynomial-time algorithms are known.

  \paragraph{Results:} 
  Here we show that the output maximization problem in chemical reaction
  networks is NP-complete. This statement remains true even if all
  reactions are monomolecular or bimolecular and if only a single molecular
  species is used as influx. As a corollary we show, furthermore, that the
  detection of autocatalytic species, i.e., types that can only be produced
  from the influx material when they are present in the initial reaction
  mixture, is an NP-complete computational problem.

  \paragraph{Conclusions:}
  Hardness results on combinatorial problems and optimization problems are
  important to guide the development of computational tools for the
  analysis of metabolic networks in particular and chemical reaction
  networks in general. Our results indicate that efficient heuristics and
  approximate algorithms need to be employed for the analysis of large
  chemical networks since even conceptually simple flow problems are
  provably intractable.
\end{abstract}

\ifthenelse{\boolean{publ}}{\begin{multicols}{2}}{}

%%%%%%%%%%%%%%%%%%%%%%%%%%%%%%%%%%%%%%%%%%%%%%%%%%%%%%%%%%%%%%%%%%%%%%%%%%%%

\section*{Background}

\begin{figure*}
\begin{center}
\includegraphics[width=\textwidth]{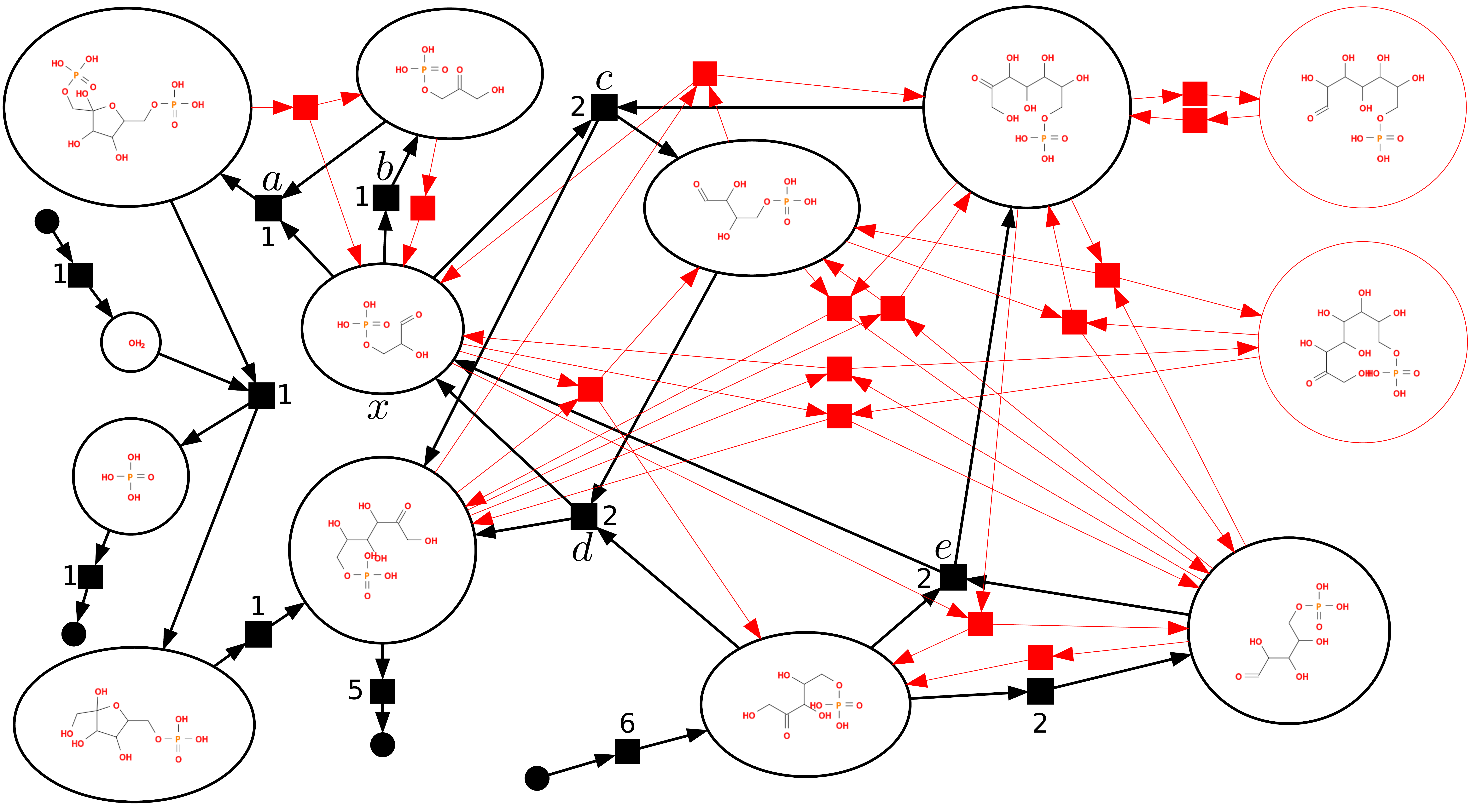}
\end{center}

\caption{Flow optimization in the pentose-phosphate reaction network.  Only
  a small part of the chemical space is shown. We allow influx of water
  $H_2O$ and ribulose-5-phosphate to generate glucose-6-phosphate as
  output. Phosphate is produced as waste product. An optimal solution is
  shown in black, using 6 ribulose-5-phosphate molecules to produce 5
  glucose-6-phosphate molecules. The values of the flow $f(\,.\,)$ is
  indicated for each hyperedge (black square), e.g., $f(a)=1$, $f(b)=1$,
  $f(c)=2$, $f(d)=2$, $f(e)=2$. At each node (except the unlabelled input
  and output nodes) the influx and outflux is balanced.  For example, at
  node $x$ (glycerol-3-phosphate), we have $f(d)+f(e)=4=f(a)+f(b)+f(c)$.}
\label{fig:pentose}
\end{figure*}

Networks of chemical reactions lie at the heart of ``systems approaches''
in chemistry and biology. After all, metabolic networks are merely
collections of chemical reactions entrenched by enzymes that favor some
possible reactions over physiologically undesirable side reactions.  A
detailed understanding of their aggregate properties thus is a prerequisite
to efficiently manipulating them in technical applications such as
metabolic engineering and at the same time form the basis for deeper
explorations into their evolution. Due to the size of reaction networks of
practical interest, efficient algorithms are required for their
investigation.

Chemical reaction networks cannot be modeled appropriately as graphs
despite the many attempts in this direction \cite{Bernal:11}. Instead, they
are canonically specified by their stoichiometric matrix $\mathbf{S}$,
augmented by information on catalysts. Equivalently, a collection of
chemical reactions on a given set of compounds forms a directed
(multi)-hypergraph \cite{Zeigarnik:00a}. As a consequence, most of
computational problems associated with chemical reaction networks cannot be
reformulated as well-studied graph problems and hence require the
development of a dedicated theory and corresponding algorithmic
approaches. Mathematical structures similar to the directed hypergraphs
arising in chemistry were also explored in a theoretical economics setting
\cite{Gallo:98,Ausiello:01}.

Two complementary approaches to analyzing chemical reaction networks have
been developed mostly in the context of analyzing and manipulating
metabolisms. Flux Balance Analysis (FBA) is concerned with the distribution
of steady-state reaction fluxes that optimize a biological objective
function such as biomass or ATP production \cite{Kauffman:03}. The
objective of metabolic design is to manipulate fluxes through a metabolic
networks so as to maximize the production of a (commercially important)
substance \cite{Hatzimanikatis:98}. More details on the structure of a
(metabolic) reaction network, on the other hand, is obtained my means of
elementary mode analysis \cite{Schuster:94}. Both approaches are concerned
with stationary mass flows through the network, mathematically given as
solution of $\mathbf{S}\vec v$, subject to the condition that flux $v_i$
through every reaction is non-negative. The elementary flux modes (EFMs)
are the extremal rays of this convex cone $\mathfrak{C}$ and can be
interpreted as a formalization of the concept of a ``biochemical pathway''
\cite{Schuster:00,Klamt:03}. FBA adds a (typically linear) objective
function to be optimized over $\mathfrak{C}$.  A major drawback of
EFM-based approaches is the combinatorial explosion of EFMs in large
networks \cite{Klamt:02} and the fact that the knowledge of EFMs does not
directly elucidate the metabolic capabilities of the given network.  An
interesting recent approach thus combines FBA with the computation of a
subset of EFMs using a greedy-like procedure \cite{Ip:11}.

Over the last years, there has been increasing interest in the
computational complexity of questions related to EFMs. For example, an
elementary flux mode can be found and counted in polynomial time
\cite{Acuna:09}.  In contrast, the question whether there is a
``futile cycle'', i.e., an EFM without input or output (equivalently,
a sub-hypergraph in which in-degree and out-degree balance for all
vertices \cite{Zeigarnik:00a}), is NP-complete
\cite{Ozturan:08}. Similarly, finding EMFs that contain two prescribed
reactions is NP-hard \cite{Acuna:10}.  A collection of reactions is a
reaction cut set for a given reaction if, after removing the cut set,
the network contains no longer an EFM containing the target reaction
\cite{Klamt:04,Klamt:06}. The problem of finding minimum cardinality
reaction cut sets is also NP-complete \cite{Acuna:09}.  The complexity
of enumerating all EFMs is still unknown \cite{Acuna:10}. In
\cite{Pitkaenen:05}, the problem of finding a shortest metabolic
pathway connecting a set of source metabolites with a desired product
is shown to be NP-hard even if stoichiometric coefficients are
neglected.

An alternative approach to analyzing the structure of chemical reaction
networks is to decompose them into a hierarchy of algebraically closed and
self-maintaining sub-networks, called chemical organizations
\cite{Kaleta:06,Centler:08,Kaleta:09,Benkoe:09a}. As shown in
\cite{Centler:08}, it is also an NP-hard problem to determine whether there
is a a given reaction network contains a non-trivial organization.

In this contribution we focus on a class of computational problems in
chemical network analysis that involve questions relating to both pathways
and organizational aspects. The problem of of maximizing production of a
desired collection of output species (rather minimizing cardinality of
reaction sets) is central to metabolic engineering \cite{Domach:04}, see
Figure~\ref{fig:pentose} for an example. In contrast to flow problems on
simple graphs \cite{Ahuja:93}, we show here that hypergraph versions
describing fluxes in chemical reaction networks are computationally
hard. As a computational problem, this flow maximization problem is closely
related to the issue of finding autocatalytic intermediates in a reaction
network. The latter problem has received considerable attention in recent
years since such ``metabolic replicators'' are universally found in
present-day metabolic networks and and likely represent their ancient
ancestral cores \cite{Kun:08}. We show here that detection of autocatalysts
is NP-hard in its general version, although a related problem in the
setting of replicator-like networks admits a polynomial-time solution
\cite{Hordijk:04}.

\section*{Result: NP-hardness}

\subsection*{Definitions}

In the following paragraphs we formally introduce chemical reaction
networks. We emphasize that our setup is the same as in the literature on
flux analysis; we have opted, however, for a somewhat different notation
that is closer to the conventions commonly used in graph theory as this
makes the subsequent discussion more concise.

A \emph{chemical reaction network} (CRN) is represented a directed
multi-hypergraph $G(V,E)$ consisting of a vertex set $V$, the compounds,
and a set $E$ of directed hyper-edges encoding the reactions
\cite{Zeigarnik:00a}. Each reaction $e\in E$ is a pair $(e^-, e^+)$ of
multisets $e^-,e^+\subseteq V$ of compounds, denoting the educts and
products of the reaction $e$. The stoichiometric coefficients $s_{x,e^-}$
and $s_{x,e^+}$ are represented by the multiplicity of the compounds in the
multisets.  For instance, the hyperedge encoding
$$C_2H_2 + 2H_2O \to (CH_2OH)_2$$
reads 
$$(\{ C_2H_2, H_2O, H_2O\},\{ (CH_2OH)_2 \})$$
Reversible reactions are encoded by a pair of forward and backward
reaction. The entries of the stoichiometric matrix are recovered as
$\mathbf{S}_{x,e}=s_{x,e^+}-s_{x,e^-}$.

In addition to the ordinary reactions like the one above, CRNs also contain
pseudo-reactions $E'$ representing influx and outflux of compounds of the
form $e_{in(x)}=(\{x_{in}\},\{x\})$ and $e_{out(x)}=(\{x\},\{x_{out}\})$
where $x_{in}$ and $x_{out}$ refer to external reservoirs. These are
additional vertices $V'$ distinct from $V$. These pseudoreactions feed the
CRN and remove ``waste products'' and extract a desired output. In
particular, the $x_{in},y_{out}\in V'$ do not take part in any other
reaction.

A flow on the directed hypergraph $G$ is a function $f:E\cup
E'\to\mathbb{N}_0$ such that, for each compound $x\in V$, the condition
\begin{equation}
  \label{eq:balance}
  \sum_{e\in E\cup E'}  f(e) \left( s_{x,e^-} - s_{x,e^+}\right) = 0  
\end{equation}
is satisfied. This condition enforce that the total production and the
total consumption of $x$ is balanced, i.e., the CRN is in a stationary
state.  The total consumption of an input material $x$ is therefore
\begin{equation}
  f(e_{in(x)}) = \sum_{e\in E} f(e)(s_{x,e^-} - s_{x,e^+})
\end{equation}
and the total outflux of a product is 
\begin{equation}
  f(e_{out(x)}) = \sum_{e\in E} f(e)(s_{x,e^+} - s_{x,e^-})
\end{equation}
We say that a species $x$ \emph{is produced} in a network if
$f(e_{out(x)})>0$.

Note that this definition of $f$ naturally generalized the definition of an
(integer) flow on a directed graph with source $x_{in}$ and target
$y_{out}$, see e.g.\ \cite{Ahuja:93}. In \cite{Cambini:97}, a
generalization of equ.(\ref{eq:balance}), although restricted to
hypergraphs with $|e^+|=1$, is considered, where the flows add up to a
vertex-dependent demand term rather than to zero.  In contrast to the usual
setting of flow problems, we have a non-trivial restriction on the capacity
only for the input edge(s), while the values of $f$ are unrestricted for
all other hyperedges.

\subsection*{Formulation of the problems}

\par\noindent\problem{MAX-CRN-Output} Given a chemical reaction network 
with $n$ nodes, of which any subset may have influx or outflux, find a flow
$f$ that maximizes the outflow $f(e_{out(y)})$ to a specified output node
$y_{out}$.

\par\noindent\problem{MAX-CRN($d$)-Output} Given a chemical reaction 
network with $n$ nodes reactions (hyperedges) with in-degree and out-degree
at most $d$, where any subset of vertices may have influx or outflux, find
a flow $f$ that maximizes the outflow $f(e_{out(y)})$ to a specified output
node $y_{out}$.

\par\noindent\problem{MAX-CRN($d$)-Output-1} Given a chemical reaction 
network with $n$ nodes, reactions (hyperedges) with in-degree and
out-degree at most $d$, and a single vertex with influx where any subset of
vertices may have outflux, find a flow $f$ that maximizes the outflow
$f(e_{out(y)})$ to a specified output node $y_{out}$.

\par\noindent\problem{Autocata} Given a chemical reaction network with 
$n$ nodes and one or more input sources, determine whether there is a
source node $x$ such that: 
\begin{enumerate}
\item $x$ cannot be produced from all other source molecules, i.e.,
  for all flows $f$, $f(e_{in(x)})=0$ implies $f(e_{out(x)})=0$;
  and
\item $x$ can be produced in a quantity that is larger than its 
  inflow, i.e., there is a flow $f$ such that 
  $f(e_{out(x)})>f(e_{in(x)})>0$.
\end{enumerate}

\subsection*{Outline}

Formally, NP-completeness is defined for decision problems
\cite{Garey:79}. Optimization problems can be converted into decision
problems by asking whether they admit a solution that is at least as good
as some value. By abuse of language, it therefore makes sense to speak of
an ``NP-complete optimization problem'' instead of using the phrase ``the
decision problem corresponding to our optimization problem is
NP-complete''.

The basic idea of proving that problem $\mathfrak{X}$ is NP-complete is to
find a so-called \emph{reduction} $\rho$ from another problem
$\mathfrak{P}$ that is already known to be NP-complete. The reduction
$\rho$ is an algorithm \emph{with polynomial runtime} that converts any
given instance of $\mathfrak{P}$ into an instance of $\mathfrak{X}$. An
efficient (i.e., polynomial time) algorithm to solve (all instances of)
$\mathfrak{X}$, therefore would also provide an efficient solution for
every instance $P\in \mathfrak{P}$ by simply reducing $P$ to
$\rho(P)\in\mathfrak{X}$ then solving $\rho(P)$. Hence we can conclude that
$\mathfrak{X}$ is a hard problem when a known hard problem $\mathfrak{P}$
can be reduced to it. 

\begin{figure*}
\begin{tabular}{ccc}
\begin{minipage}{0.55\textwidth}
\begin{center}
\includegraphics[width=\textwidth]{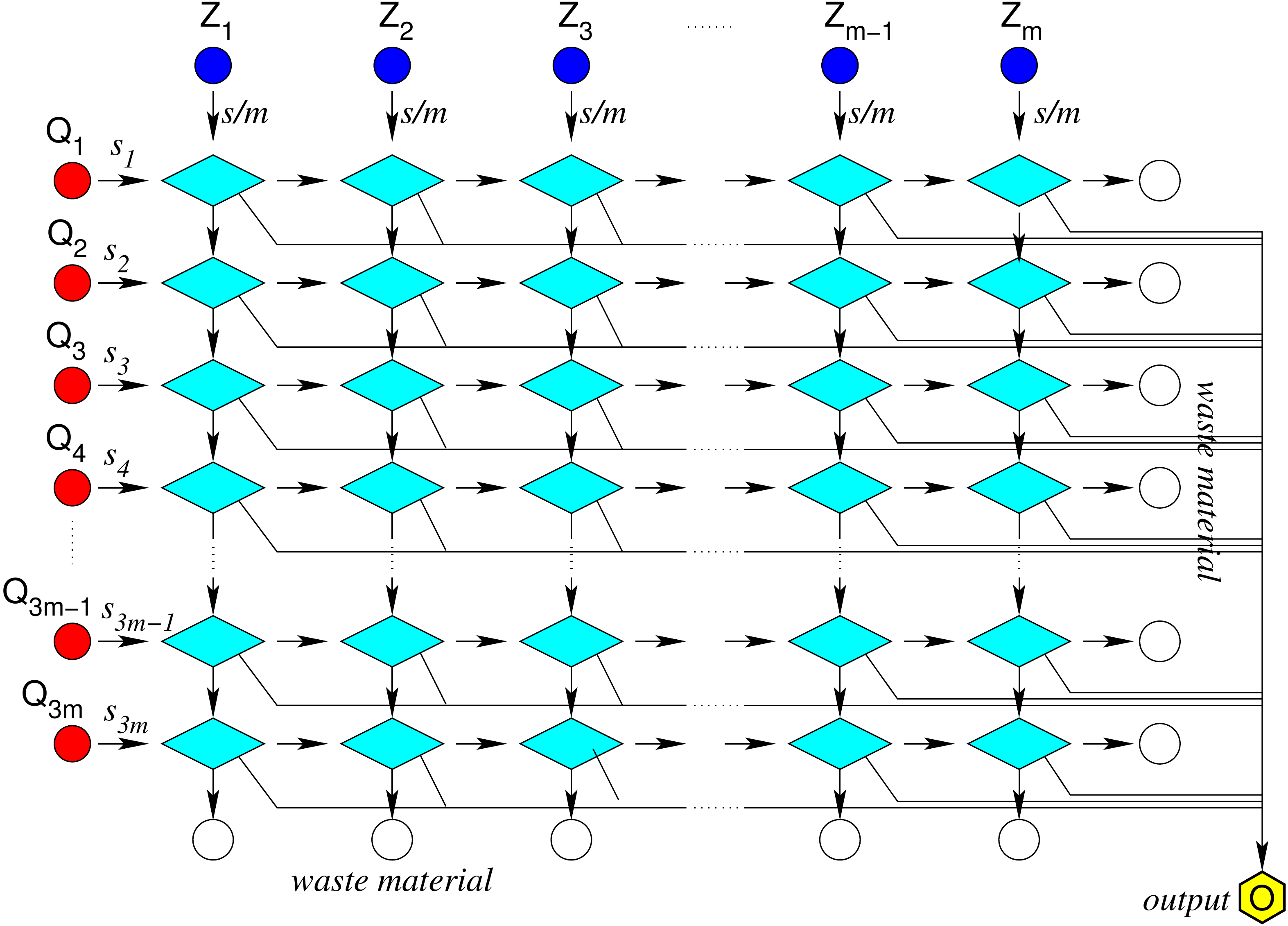}
\end{center}
\end{minipage} && 
\begin{minipage}{0.35\textwidth}
\begin{center}
\includegraphics[width=\textwidth]{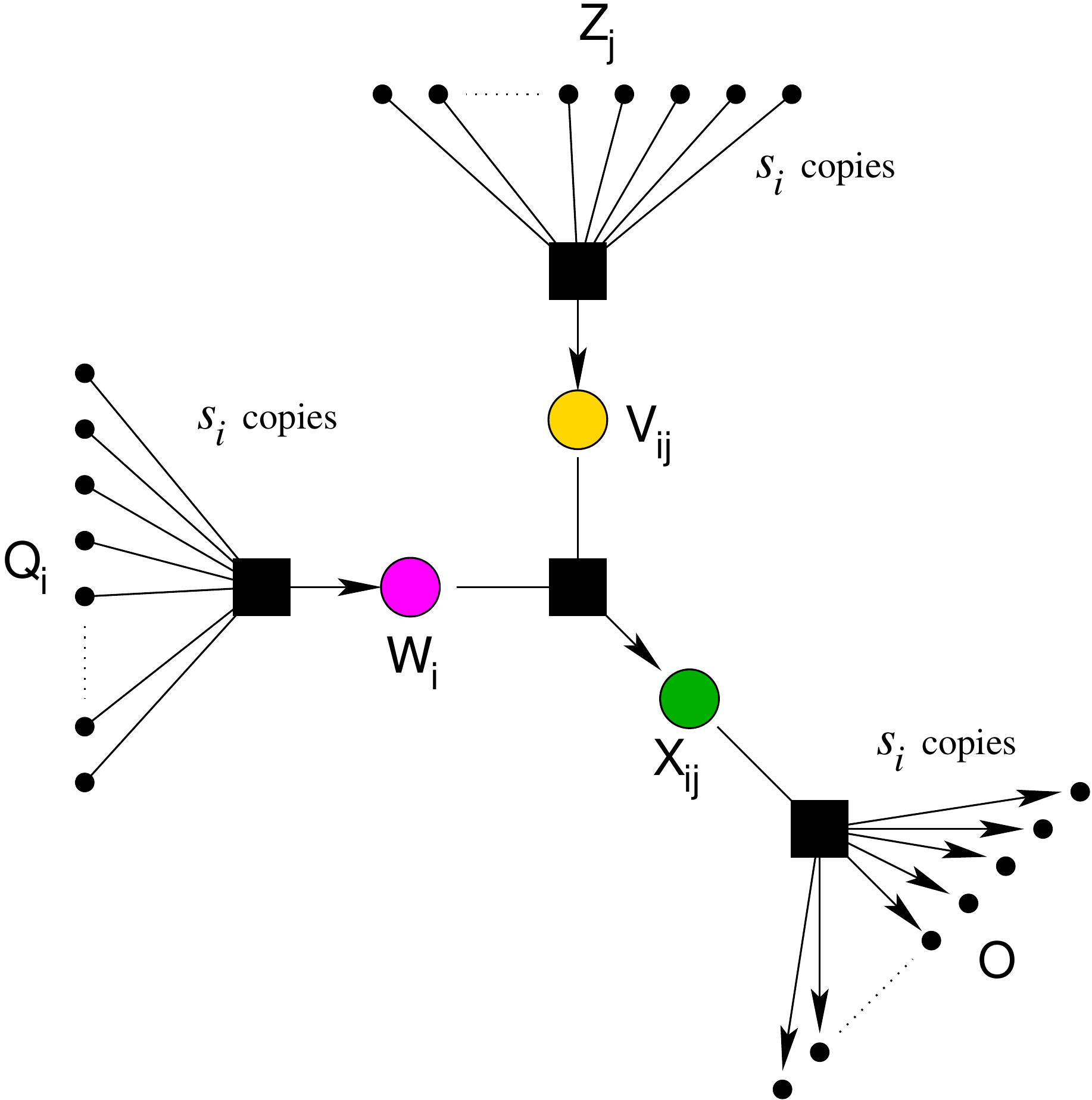}
\end{center}
\end{minipage}\\
(A) & & (B) \\
\end{tabular}
\caption{Construction of a CRN from a given instance of
  \problem{3PART}. (A) In the first step, an intermediate network
  consisting of input nodes, switch nodes (green diamonds), and waste nodes
  (open circles), and a single output sink (hexagon) is constructed. The
  input is encoded as capacity constraint on the l.h.s.\ input nodes
  (corresponding to the input numbers $s_i$ of \problem{3PART} and on the
  $m$ top nodes (corresponding to $1/m$ of the sum of the inputs). A
  solution of \problem{3PART} corresponds to a flow through this network
  that transport $\sum_i s_i$ to the output sink. (B) In the second step,
  each switch node is replaced by reaction network that which admits a
  non-zero flow only if $s_i$ copies of $Q_i$ and $Z_j$ are available. The
  reaction then produces $s_i$ copies of the output molecule $O$. Note that
  the ``drainage reactions'' as not shown in panel (B). These channel the
  $Q_j$ and $Z_j$ input material directly to the ``waste material'' sink
  whenever the reaction networks inside the switch node receives
  insufficient input to produce both $W_i$ and $V_{ij}$.}
\label{fig:lattice}
\vspace*{1cm}
\end{figure*}

In this section we devise a procedure that reduces every instance of the
so-called 3-partition problem to a CRN with a single output pseudo-reaction
in such a way that solving the output maximization problem for the CRN also
solves the 3-partition problem. Thus optimizing output in CRNs is at least
as hard as solving 3-partition. The same basic construction is then
modified to show that the CRN can be built in such a way that all reactions
are monomolecular or bimolecular. We then employ the same construction to
show that problem remains hard even if only a single source is provided. A
simple modification finally establishes the hardness result for finding
autocatalytic compounds.

\subsection*{3-Partition}

The 3-partition problem (\problem{3PART}) consists in deciding whether
a given multiset of $n=3m$ integers $s_i$, $i=1,\dots,3m$ can be
partitioned into triples that all have the same sum. This problem is
one of the most famous strongly NP-complete problems, i.e., it stays
NP-complete even when the numbers in the input instance are given in
unary encoding \cite{Garey:75}, i.e., their values grows not faster
than a polynomial in the problem size $n$. This remains true when the
$s_i$ are distinct \cite{Hulett:08}. If $B$ denotes the desired sum of
each subset then \problem{3PART} remains strongly NP-complete
even if for every integer $B/4 < s_i < B/2$ holds. 

\subsection*{Basic Construction}

Given an instance of \problem{3PART} we construct the associated CRN in a
step-wise fashion. The first step is a lattice-like labeled graph,
Figure~\ref{fig:lattice}(A), that consists of one input node corresponding to
each $s_i$, $m$ auxiliary nodes $Z_j$, each of which has an
influx of $(1/m)\sum_i s_i = s/m$, an output sink node, $3m\times m$ switch
nodes, $3m$ waste nodes at the right and $m$ waste nodes at the bottom.
These switch nodes have two inputs $l$ from the left and $u$ from above,
and three outputs $r$ towards the right, $d$ downwards, and $o$ into the
output channel. Each of the switch nodes can be in one of two distinct
states: either it
\begin{enumerate}
\item[off] The node transmits all its left input to right \textbf{and} all
  its input from above downwards, no flow is then diverted towards the
  output, i.e., $r=l$, $d=u$, $o=0$; or
\item[on] The node consumes its entire input from the left (and thus
  transmits nothing to the right), at the same time uses up a corresponding
  amount of the input from above, and diverts a corresponding amount towards
  the output, i.e., $r=0$, $d=u-l$, $o=l$.
\end{enumerate}
All flux along the output channel is collected in the output node, i.e.,
given a particular state of the switch nodes, the flux into the output node
is the sum of the fluxes consumed from the left. 

\begin{lemma}
  An assignment of ``on'' and ``off'' to the $3m\times m$ switch nodes is a
  solution of the original \problem{3PART} problem if and only if the total
  flow in the output node $O$ equals the maximally possible value $s =
  \sum_i s_i$.
\end{lemma}
\begin{proof}
  Consider the CRN in Figure~\ref{fig:lattice} with $3m \times m$ switch
  nodes. Each column corresponds to one of the $m$ desired subsets of
  the underlying instance of \problem{3PART}, each row corresponds to
  one the $3m$ integer values $s_k$. Note that any assignment of
  ``on'' and ``off'' to switch nodes will split the overall horizontal
  as well as the overall vertical inflow into two parts: a part
  directed to waste material and an output part directed to node
  $O$. Let $w_H$ (resp.\ $w_V$) be the overall horizontally (resp.\
  vertically) produced waste. For any assignment of ``on'' and ``off''
  states to switch nodes $s = f(e_{out(O)}) + w_H = f(e_{out(O)}) + w_V$
  is invariant. Obviously, if $w_H = w_V = 0$, then the outflow
  $f(e_{out(O)})$ to node $O$ is maximal. Furthermore note that at
  most one switch can be in ``on'' state in each row.

  Consider an assignment of ``on'' and ``off'' to the switch nodes that
  corresponds to a solution of the original \problem{3PART} problem. Thus 
  exactly $3m$ switch nodes are in mode ``on'' (three per column
  and one per row). As one switch node per row $i$ is in mode ``on'',
  the outflux $s_i$ of node $Q_i$ flows to output node $O$ and the
  waste produced horizontally in row $i$ is $0$. As this is true for
  all rows, $w_H=w_V=0$ holds and the total flow in the output node
  $O$ is $s$ which is maximal.

  Assume that the flow in the output node is the maximal possible
  value $s$, and therefore $w_H=w_V=0$ holds. This implies that
  exactly one switch node per row needs to be in mode ``on''. As we
  can assume $s/(4m) < s_i <s/(2m)$ exactly 3 switch nodes per column
  need to be in state ``on''.  The overall assignment is therefore a
  solution to the original \problem{3PART} problem.
\end{proof}

Of course, the intermediate network in Figure~\ref{fig:lattice}(A) is not
(yet) an proper CRN. To achieve this goal, we have to replace the switch
nodes by hypergraphs that implement the high-level rule governing their
behavior. 

\subsection*{Implementing switch-nodes}

Suppose the molecules emitted from the $3m$ input nodes are all of
different types $Q_i$, and distinguish the $m$ types of inputs from above
as $Z_j$. Then the switch node $(i,j)$ must implement a net 
reaction of the form
\begin{equation}
  s_i Q_i + s_i Z_j \to s_i O 
\end{equation}
where $O$ is the type of the output molecule. This net reaction can be 
split into four subsequent reactions:
\begin{equation}
  \begin{split}
    s_i Q_i & \to W_i \\
    s_i Z_j & \to V_{ij} \\
    V_{ij} + W_i & \to X_{ij} \\
    X_{ij} & \to s_i O 
  \end{split}
  \label{eq:switch}
\end{equation}
We see that the switch node $(i,j)$ can be in the ``on''-state only if it
received at least $s_i$ copies of the input from the left and a matching
number of input molecules from above. A graphical description of this
partial network is shown in Figure~\ref{fig:lattice}(B).  Since the input
from the left is limited to $s_i$ copies of $Q_i$, either none or a single
molecule of the intermediate $X_{ij}$ is produced, depending on whether
$(i,j)$ is on or not. Clearly, for each $i$, only a single one of the
switches $(i,j)$ can be ``on''.

Note that equ.(\ref{eq:switch}) already provides the necessary device
to complete the proof. If we insist that the CRN may use at most
bi-molecular reactions, we have to find a way to implement the
reactions $s_i Q_i \to W_i$ and $X_{ij} \to s_i O$ by more restricted
elementary reactions. This will the topic of the following
section. According to equ.(\ref{eq:switch}) each diamond node is
replaced by $3(s_i+1)$ vertices, so that the entire network has $6m +
2m + 1 + m\sum_{i=1}^{3m} 3(s_i+1) = 8m + 3sm + 3m^2 + 1$ nodes.
Thus, all instances of \problem{3PART} for which $s=s(m)$ is
polynomially bounded in $m$ can be reduced to a maximum output problem
on an equivalent CRN. We explicitly use the fact that \problem{3PART}
is \emph{strongly} NP-complete: we need that $m$ is polynomially
bounded by the network size $n$ to ensure that $s$, and thus the
reduction to \problem{3PART}, remains polynomial. We know the maximal
outflux of the CRN and can therefore use a simple guess-and-check
argument to show that \problem{MAX-CRN-Output} is in NP. Our
discussion thus establishes
\begin{theorem}
  \problem{MAX-CRN-Output} is strongly NP-complete when the number of
  inputs into the CRN and number of educts in a chemical reaction is
  unrestricted.
\end{theorem}

We remark the our CRNs need to have at least two output nodes, one for
the desired product and one to collect all waste products.  

\subsection*{Restriction to bi-molecular reactions}

In this section we show that the problem does not become easier when the
CRN has only a single input and all reactions are bi-molecular. To this end
we further refine the reactions $s_i Q_i\to W_i$, $X_{ij} \to s_i O$. We
will make use of two specialized types of edges that can be implemented
by bi-molecular reactions.

The first type of edge merges exactly $k$ identical molecules into $1$
molecule (the corresponding edges will be referred to as
merge-edges). The second type of edge expands one molecule to exactly
$k$ identical molecules (expansion-edges). We first focus on a
specific type of merge- and expansion-edges: merge-edges of type $(2^u
\to 1)$ can easily be implemented by $u$ subsequent reactions $f^i,
i=1,\ldots,u $ that iteratively create (double-sized) molecules out of
2 identical molecules. Formally, let $I=X_1$ and $O=X_{u+1}$ then
$f^i$ is defined by
\begin{equation}
  2 X_i \to X_{i+1},
\end{equation}
and the corresponding flow is chosen to be $f^i(\{X_i,X_{i+1}\})
:=2^{u-i}$. Symmetrically, expansion-edges of type $(1 \to 2^u)$ can
be implemented by $u$ subsequent reactions that split molecules
repeatedly into two equal molecules. These $(2^u \to 1)$-merge-edges
(resp. $(1 \to 2^u)$-expansion-edges) will in the following be used to
implement the generalized merge- and expansion-edges.

Let $b_{m-1} b_{m-2} \ldots b_{0}$ be the binary representation of $k>0$
with $m=\lfloor \log k \rfloor +1$, and let $B=\{i_1, i_2, \ldots,
i_r\}$ be the indices of all non-zero bits, i.e $i \in B$ with
$b_i=1$. The underlying idea for the merging of $k$ molecules of type
$I$ into 1 molecule of type $O$ is to split the outflow $k$ of $I$
into $r$ individual flows, i.e.  $k=\sum_{j=1}^r 2^{i_j-1}$. We remark
that this representation is unique.  These flows of quantity
$2^{i_j-1}, j=1 \ldots r$ are then individually reduced to flows of
size $1$.  The resulting $r$ flows of quantity 1 are then all merged
to a flow of one molecule of quantity $1$. The implementation of
generalized merge-edges is depicted in Figure
\ref{fig:edges}(A). Expansion-edges that expand the flow of one
molecule of quantity $1$ to a flow of one molecule of quantity $k$ can
be implemented analogously.  First, a flow of quantity $1$ of one
molecule is changed into $r$ flows of quantity $1$, then these $r$
flows are expanded to $r$ flows of quantity $2^{i_j-1}, j=1,\ldots r$,
and then these flows are iteratively summed up. The details are
depicted in Figure \ref{fig:edges}(B).  Clearly, merge and expansion
edges can be employed for the refinement of reactions $s_i Q_i\to
W_i$, $X_{ij} \to s_i O$ in equ.\eqref{eq:switch}.  The number of
additional edges and nodes to implement a $(k\to 1)$ merge-edge is
$O(\log^2 k)$, as there are $O(\log k)$ flows after the split into
individual flows, and each individual flow employs $O(\log k)$ edges
for the $(k \to 1)$ merge (with $k$ being a power of 2). Symmetrically
a $(1\to k)$ expansion-edge uses $O(\log^2 k)$ bi-molecular edges and
additional compounds. Based on this polynomial extension and as all
merge and expansion reactions are bi-molecular, we have the following
\begin{corollary}
  \problem{MAX-CRN(2)-Output} is strongly NP-complete.
\end{corollary}

\begin{figure*}
\begin{tabular}{ccc}
\begin{minipage}{0.44\textwidth}
\includegraphics[width=\textwidth]{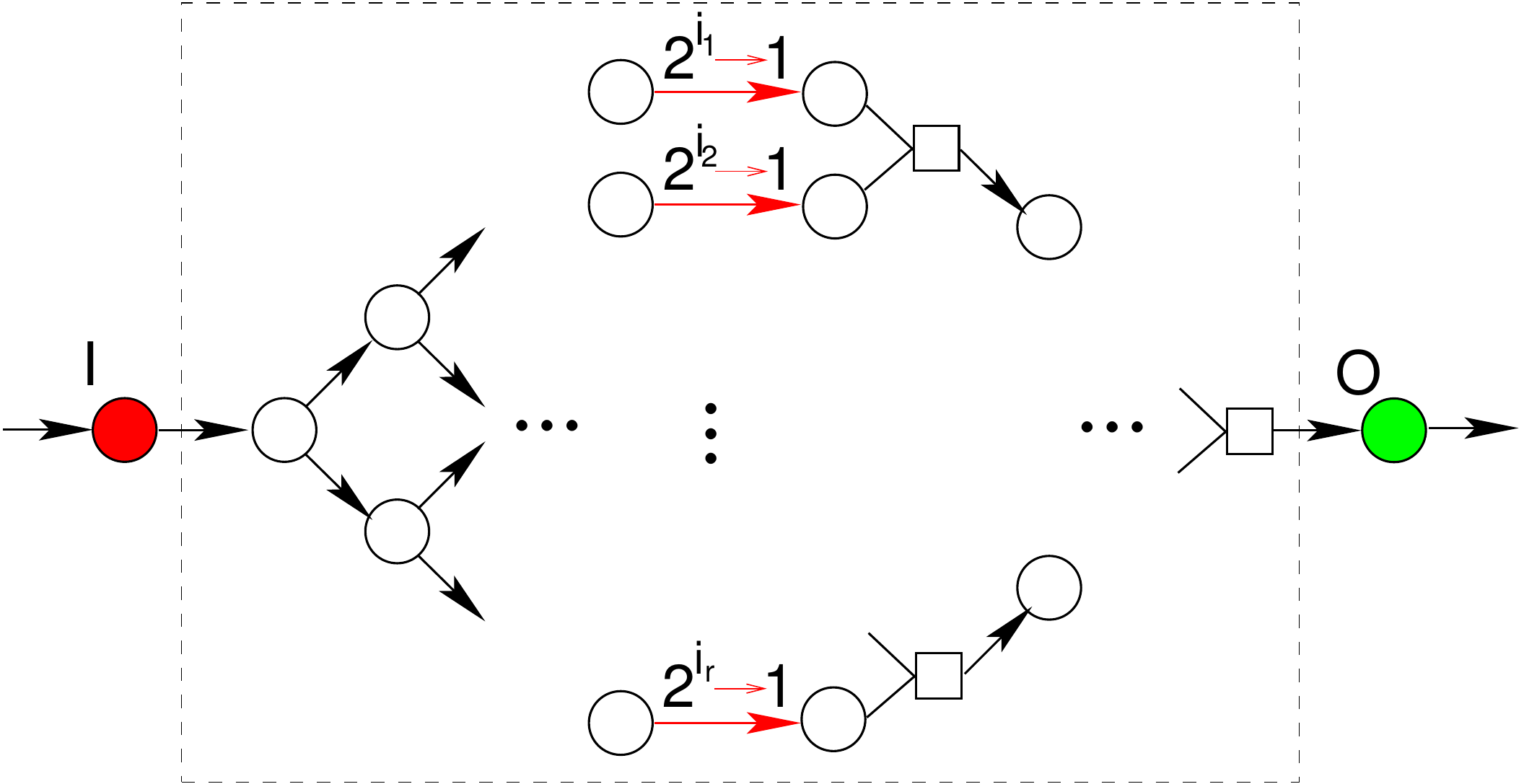}
\end{minipage} &&
\begin{minipage}{0.50\textwidth}
\includegraphics[width=\textwidth]{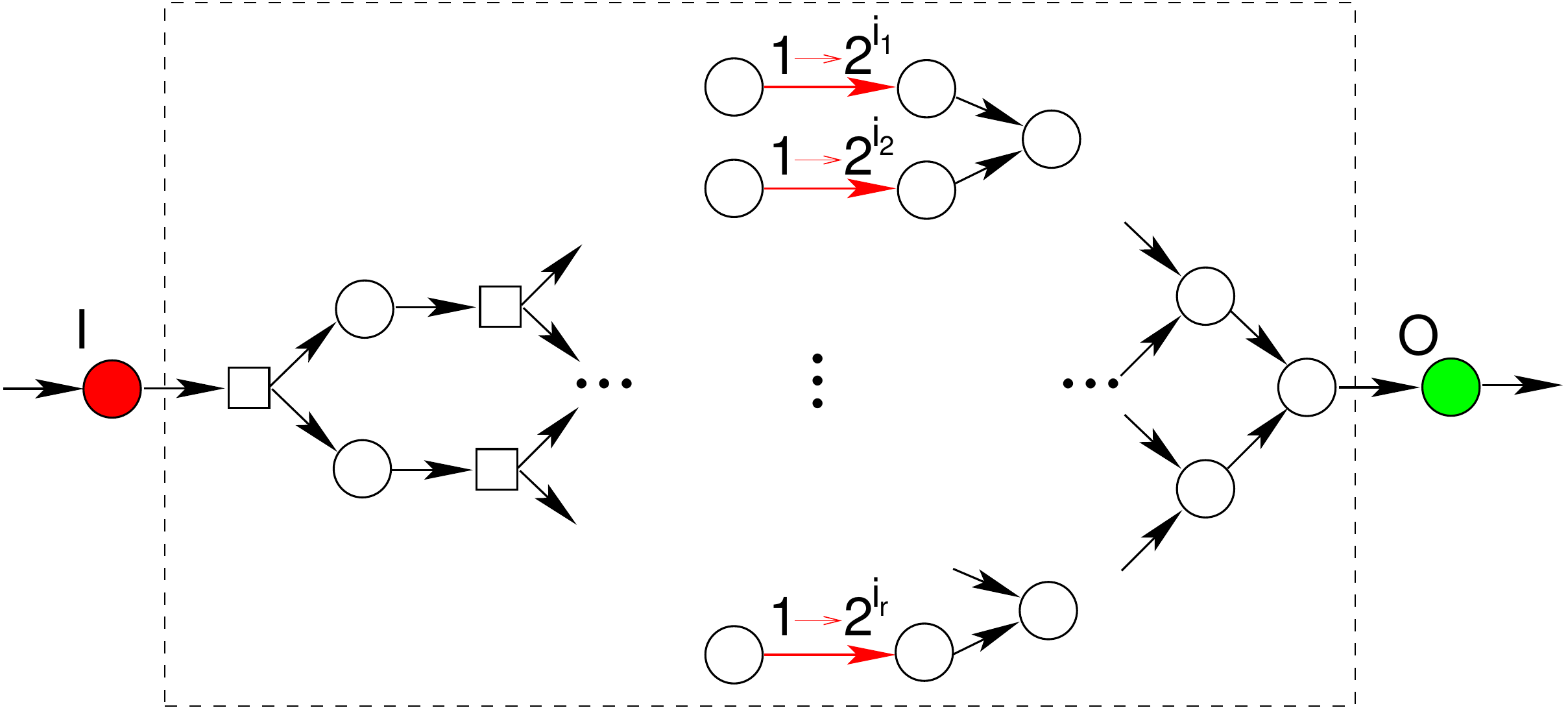}
\end{minipage}\\
(A) && (B)
\end{tabular}
\caption{Consider the binary representation $b_{m-1} b_{m-2} \ldots b_{0}$ of  
  $k>0$ with $m=\lfloor \log k \rfloor +1$. Let $B=\{i_1, i_2, \ldots,
  i_r \}$ be the indices of all non-zero bits, i.e., $i \in B$ with
  $b_i=1$. (A) Implementation of a $(k\to 1)$ merge-edge. (B)
  Implementation of a $(1\to k)$ expansion-edge. The red edges indicate 
  $(2^i\to1)$ merges and $(1\to 2^i)$ expansions, respectively.}
\label{fig:edges}
\end{figure*}

\subsection*{Restriction to a single input} 
To show that \problem{MAX-CRN-Output} is NP-complete even if we have a
single input only, we require an additional edge type that is
implemented by connecting a $(k\to 1)$-merge-edge and a $(1 \to
k)$-expansion edge in series. Such an edge ensures that exactly $k$
(or exactly a multiplicity of $k$) input molecules react to the same
number of output molecules. We will refer to these edges as
$(k)$-force-flow-edges. Note, that such edges do not change the
quantity of a flow. The number of additional edges and nodes required
to implement a $(k)$-force-flow edge is $O(\log^2 k)$.

So far we assumed input nodes $Q_i$ with corresponding influx $s_i$, $i = 1
\ldots, 3m$, plus the $m$ additional input nodes $Z_1, \ldots, Z_m$ with
influx $s = (1/m)\sum_i s_i$ each. In the following we will describe how to
extend the construction of the CRN based on an instance of the
\problem{3PART} problem (\emph{cmp.}\ Figure \ref{fig:lattice}) such that
there is only a single input node. Note that all $s_i$, $m$, and the influx
to nodes $Z_i$ are defined by the given \problem{3PART} instance.

\paragraph{Influx to nodes $Q_i$:}
In the extended CRN the nodes $Q_i$ will be internal nodes with influx
$s_i$. In order to achieve this we will add a single input node $Q$ with
influx $s'$, where $s'$ is the integer representation of the concatenation of
the $r$-bit binary representation of all $s_i$, i.e., 
\begin{equation}
  s'=\sum_{i=1}^{3m} s_i \times 2^{r (i-1)}, \text{~with~} 
  r = \max \{ \lfloor \log s_i \rfloor \} +1
\end{equation} 
Attached to node $Q$ will be a subnetwork that splits the flux $s'$
into the fluxes $s_1, \ldots, s_{3m}$ by iteratively using the last
$r$ bit of the remaining flux as influx to a node $Q_i$, and then
divide the remaining flux by $2^r$. The hypergraph structure to
implement this with bi-molecular reactions only is depicted in Figure
\ref{fig:split}. All dashed lines with red rectangles indicate
force-flow-edges (the number in the rectangle indicates the enforced
flow), all red edges with open arrowheads indicate merge- or
expansion- edges. To enforce that exactly (and not a multiplicity) of
$s_i$ molecules flow towards node $Q_i$, the flow downwards needs to
be maximized. This is done by introducing an additional outflux node:
the flux of quantity $s_{3m}\geq 1$ towards $O'$ is multiplied by a
factor $c$, such that the additional overall non-waste outflux to $O'$
dominates any other non-waste outflux. This can be ensured by choosing
the factor $c$ as the maximal possible influx to $Q$, i.e. $c= 2^{r
  \times 3m}-1$ (the binary representation of $c$ has $r \times 3m$
bit all set to 1). The number of additional edges and nodes is
polynomially bound and the overall outflux of the extended network is
then $s_{3m} \times c + \sum_i{s_i}$. As all outflux can be easily
merged in a binary fashion as applied in the definition of
expansion-edges, the resulting CRN has only a single input node and a
single non-waste output node.

\begin{figure*}
\begin{tabular}{lcr}
\begin{minipage}{0.63\textwidth}
\begin{center}
\input{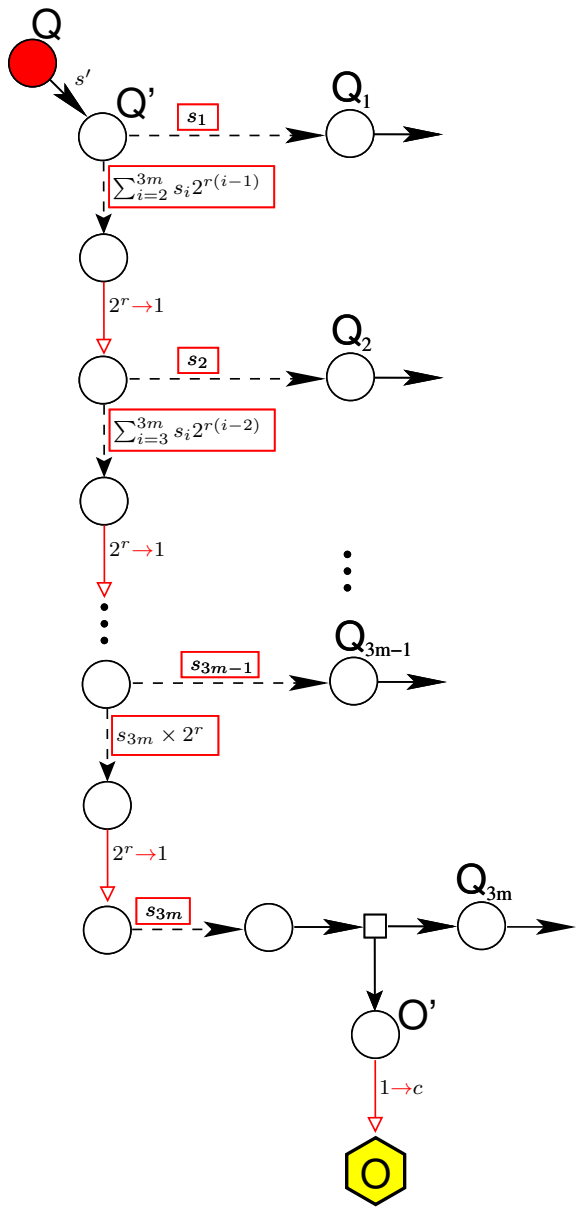}
\end{center}
\end{minipage} & \quad & 
\begin{minipage}{0.3\textwidth}
\caption{Splitting the single influx $s'$ to node $Q'$ such that the influxes
  to the internal nodes $Q_i$ are $s_i$: the influx to node $Q$ is chosen
  to have the quantity $s'=\sum_{i=1}^{3m} s_i \times 2^{r (i-1)}$ with $r
  = \max \{ \lfloor \log s_i \rfloor \} +1$, i.e., $s'$ is determined by the
  concatenation of binary representation of the values $s_i$; force-flow
  edges are depicted as dashed lines labeled with the enforced quantity,
  merge- (resp. expansion-) edges are depicted as red lines with open
  arrowheads labeled the quantification of merging (resp., expansion); the
  constant $c$ for the expansion towards node $O$ is chosen such that the
  outflux in node $O$ dominates the outflux of the original lattice CRN.}
\label{fig:split}
\end{minipage}
\end{tabular}
\end{figure*}

\paragraph{Influx to nodes $Z_i$:}
In order to have nodes $Z_i$ (\emph{cmp.}\ Figure\ \ref{fig:lattice}) as
internal nodes, we split the outflux from node $Q$ of quantity $s'$ in
two fluxes of quantity $s'-1$ and $1$ (by employing force-flow-edges),
that will be directly merged again and be used as influx of quantity
$s'$ to node $Q'$. However, this simple splitting procedure gives a
flux of quantity $1$. This simple flux is easily transformed into $m$
fluxes of quantity $1$, which are then multiplied by $s/m$ using
expansion-edges, and then used as the input towards the internal nodes
$Z_i$.

Recall, that the number of nodes and edges needed for a
force-flow-edge of quantity $k$ is $O(\log^2 k)$. The number of bits for
the maximal flux on any force-flow-edge is $O(r \times 3m)$. As
\problem{3PART} is strongly NP-complete we can assume that all $s_i$
are polynomially bound in $m$, and therefore $r \in O(\log
m)$. Therefore the maximal flux on any edge is $O(2^{m \log m})$. The
number of additional nodes and edges is therefore $O(m^2\log^2 m)$ per
force-flow-edge. As the construction needs $O(m)$ additional
force-flow-edges, the overall number of additional nodes and edges is
$O(m^3 \log^2 m)$. Therefore the following corollary easily follows:
\begin{corollary}
  \problem{MAX-CRN(2)-Output-1} is NP-complete.
\end{corollary}

\subsection*{Autocatalysis} 
The NP-completeness of detecting an autocatalytic species can be shown
by expanding the CRN used for showing the NP-completeness of
\problem{MAX-CRN(2)-Output-1}. Let $O$ be the output node, where a
outflux of $s_{3m} \times c + \sum_i{s_i}$ can be detected iff the
underlying instance of \problem{3PART} is solved. We add a merge-edge
from $O$ towards an additional node $A'$ to create an outflux of
exactly $1$ from $A'$. The CRN is furthermore extended by the
following two additional reactions, where compound $A$ is an input and
an output node of the CRN.
\begin{eqnarray*}
A' + A &\to& 2B\\
B &\to& A
\end{eqnarray*}
The outflux of $A'$ is 1, if and only if 
\begin{enumerate}
\item Compound $A$ cannot be produced from all other source molecules, i.e.,
for all flows $f(e_{in(A)})=0$ implies $f(e_{out(A)})=0$, and
\item two $A$ can be produced if their is an inflow of one $A$, i.e.,
  there is a flow $f$ such that $f(e_{out(A)})>f(e_{in(A)})>0$.
\end{enumerate}
The construction of our reduction highlights the difficult part in
determining autocatalysts. This is not so much finding the autocatalytic
cycle itself but to ensure that the building blocks are provided from the
``food source'' through an in principle arbitrarily complicated
sub-network. 

\section*{Concluding Remarks}

We have shown that the flow maximization problem and the detection of
autocatalytic species in chemical reaction networks are NP-complete
computational problems. As a consequence , we cannot expect to find devise
exact algorithms for these problems that can be used efficiently on large
chemical reaction networks (unless P=NP, which is unlikely at best
\cite{Fortnow:09}). Our results match well with the observation that many
classical computational problems are hard on hypergraphs even though their
analogs for simple graphs admit efficient exact solutions. Illustrative
examples are matching problems \cite{Karp:72}, or the sparsest null space
problem for integer matrices \cite{Coleman:86}, which can be seen as the
natural generalization of the minimum cycle basis problem. As graph models
of chemical networks tend to be oversimplifications that are often of
limited use \cite{Bernal:11}, the hardness of the computational task
associated with the analysis of large reaction networks cannot be avoided.
As exact algorithms appear out of reach, it will be necessary to
systematically explore efficient approximation algorithms and heuristics
for the combinatorial problems naturally arising from Systems Chemistry.

%%%%%%%%%%%%%%%%%%%%%%%%%%%%%%%%%%%%%%%%%%%%%%%%%%%%%%%%%%%%%%%%%%%%%%%%%%%%

\section*{Authors contributions}
D.M. designed the study. All authors contributed to the results and the
writing of the manuscript and approved the submitted manuscript. 

%%%%%%%%%%%%%%%%%%%%%%%%%%%
\section*{Acknowledgements}
  \ifthenelse{\boolean{publ}}{\small}{}

  This work was supported in part by the Volkswagen Stiftung proj.\
  no.\ I/82719, and the COST-Action CM0703 ``Systems Chemistry'' and
  by the Danish Council for Independent Research, Natural Sciences.

{\ifthenelse{\boolean{publ}}{\footnotesize}{\small}
 \bibliographystyle{bmc_article}  % Style BST file
  \bibliography{NPflow} }     % Bibliography file (usually '*.bib' ) 

%%%%%%%%%%%

\ifthenelse{\boolean{publ}}{\end{multicols}}{}

%\section*{Additional Files}
%\subsection*{Additional File 1 --- if any ???}
%caption of additional file 1 

\end{bmcformat}
\end{document}